\documentclass[10pt]{article}
\usepackage{amsmath,amsthm,color}
\usepackage{amssymb}

\newtheorem{theorem}{Theorem}
\newtheorem{lemma}[theorem]{Lemma}
\newtheorem{proposition}[theorem]{Proposition}

\newtheorem*{claim}{Claim}

\usepackage{mathabx}
\makeatletter
\newcommand\incircbin
{%
  \mathpalette\@incircbin
}
\newcommand\@incircbin[2]
{%
  \mathbin%
  {%
    \ooalign{\hidewidth$#1#2$\hidewidth\crcr$#1\ovoid$}%
  }%
}
\newcommand{\oeq}{\incircbin{=}}
\makeatother

\newcommand{\RR}{\mathbb{R}}

\newcommand{\ip}[1]{\left\langle #1 \right\rangle}
\newcommand{\Etilde}{\widetilde{\EE}}
\DeclareMathOperator*{\argmax}{\mathrm{argmax}}
\DeclareMathOperator*{\argmin}{\mathrm{argmin}}

\newcommand{\diag}{\mathrm{diag}}
\newcommand{\tr}{\mathrm{tr}}
\newcommand{\var}{\mathrm{Var}}
\newcommand{\one}{\mathbf{1}}
\newcommand{\bA}{\mathbf{A}}
\newcommand{\bT}{\mathbf{T}}
\newcommand{\bW}{\mathbf{W}}
\newcommand{\EE}{\mathbb{E}}

\newcommand{\cA}{\mathcal{A}}
\newcommand{\cB}{\mathcal{B}}

\newcommand{\cF}{\mathcal{F}}
\newcommand{\cE}{\mathcal{E}}
\newcommand{\cL}{\mathcal{L}}
\newcommand{\cU}{\mathcal{U}}
\newcommand{\cV}{\mathcal{V}}

\newcommand{\ceil}[1]{\lceil #1 \rceil}

\definecolor{NYUpurple}{rgb}{0.34,0.02,0.55}

\title{Community Detection in Hypergraphs, Spiked Tensor Models, and Sum-of-Squares}

\author{
Chiheon Kim\thanks{MIT, Department of Mathematics, 77 Massachusetts Ave., Cambridge, MA 02139}\\
\texttt{chiheonk@math.mit.edu}
\and
Afonso S. Bandeira\thanks{Department of Mathematics and Center for Data Science, Courant Institute of Mathematical Sciences, New York University, New York, NY 10012. Part of this work was done while A. S. Bandeira was with the Mathematics Department at MIT and supported by NSF Grant DMS-1317308.}\\
\texttt{bandeira@cims.nyu.edu}
\and
Michel X. Goemans\thanks{MIT, Department of Mathematics, Room 2-474, 77 Massachusetts Ave., Cambridge, MA 02139. Partially supported by ONR grants N00014-14-1-0072 and N00014-17-1-2177.}\\
\texttt{goemans@math.mit.edu}
}

\begin{document}

\maketitle
\begin{abstract}

We study the problem of community detection in hypergraphs under a stochastic block model. Similarly to how the stochastic block model in graphs suggests studying spiked random matrices, our model motivates investigating statistical and computational limits of exact recovery in a certain spiked tensor model. In contrast with the matrix case, the spiked model naturally arising from community detection in hypergraphs is different from the one arising in the so-called tensor Principal Component Analysis model. We investigate the effectiveness of algorithms in the Sum-of-Squares hierarchy on these models. Interestingly, our results suggest that these two apparently similar models exhibit significantly different computational to statistical gaps.

\end{abstract}

\section{Introduction}

Community detection is a central problem in many fields of science and engineering. It has received much attention for its various applications to sociological behaviours \cite{goldenberg2010survey,fortunato2010community,
newman2002random}, protein-to-protein interactions \cite{marcotte1999detecting, chen2006detecting}, DNA 3D conformation \cite{cabreros2016detecting}, recommendation systems \cite{linden2003amazon, sahebi2011community, wu2015clustering}, and more. In many networks with community structure one may expect that the groups of nodes within the same community are more densely connected. The \emph{stochastic block model} (SBM) \cite{holland1983stochastic} is arguably the simplest model that attempts to capture such community structure.

Under the SBM, each pair of nodes is connected randomly and independently with a probability decided by the community membership of the nodes. The SBM has received much attention for its sharp phase transition behaviours  \cite{mossel2013proof}, \cite{abbe2015community}, \cite{abbe2016exact}, computational versus information-theoretic gaps \cite{chen2016statistical}, \cite{abbe2015recovering}, and as a test bed for many algorithmic approaches including semidefinite programming \cite{abbe2016exact}, \cite{hajek2016achieving}, spectral methods \cite{massoulie2014community}, \cite{Vu2014} and belief-propagation \cite{abbe2015detection}. See \cite{abbe2016community} for a good survey on the subject. 

Let us illustrate a version of the SBM with two equal-sized communities. Let $y \in \{\pm 1\}^n$ be a vector indicating community membership of an even number $n$ of nodes and assume that the size of two communities are equal, i.e., $\one^T y = 0$. Let $p$ and $q$ be in $[0, 1]$ indicating the density of edges within and across communities, respectively. Under the model, a random graph $G$ is generated by connecting nodes $i$ and $j$ independently with probability $p$ if $y_i=y_j$ or with probability $q$ if $y_i \neq y_j$. Let $\widehat{y}$ be any estimator of $y$ given a sample $G$. We say that $\widehat{y}$ exactly recovers $y$ if $\widehat{y}$ is equal to $y$ or $-y$ with probability $1-o_n(1)$. In the asymptotic regime of $p = a\log n / n$ and $q = b \log n / n$ where $a > b$, it has been shown that the maximum likelihood estimator (MLE) recovers $y$ when $\sqrt{a}-\sqrt{b} > \sqrt{2}$ and the MLE fails when $\sqrt{a}-\sqrt{b} < \sqrt{2}$, showing a sharp phase transition behaviour \cite{abbe2016exact}. Moreover, it was subsequently shown in \cite{hajek2016achieving} and \cite{Bandeira2016laplacian} that the standard semidefinite programming (SDP) relaxation of the MLE achieves the optimal recovery threshold. 

A fruitful way of studying phase transitions and effectiveness of different algorithms for the stochastic block model is to consider a Gaussian analogue of the model \cite{javanmard2016phase}. Let $G$ be a graph generated by the SBM and let $A_G$ be the adjacency matrix of $G$. We have
$$
\EE A_G = \frac{p+q}{2} \one\one^T + \frac{p-q}{2} yy^T - pI.
$$
It is then useful to think $A_G$ as a perturbation of the signal $\EE A_G$ under centered noise. 

This motivates a model with Gaussian noise. Given a vector $y \in \{\pm 1\}^n$ with $\one^T y = 0$, a random matrix $T$ is generated as
$$
T_{ij} = y_i y_j + \sigma W_{ij}
$$
where $W_{ij} = W_{ji} \sim N(0,1)$ for all $i, j \in [n]$. This model is often referred to $\mathbb{Z}_2$-synchronization \cite{Bandeira2016laplacian,javanmard2016phase}. It is very closely related to the \emph{spiked Wigner model} (or spiked random matrix models in general) which has a rich mathematical theory \cite{feral2007largest,montanari2015semidefinite,
perry2016optimality}.


In many applications, however, nodes exhibit complex interactions that may not be well captured by pairwise interactions \cite{agarwal2005beyond, zhou2006learning}. One way to increase the descriptive ability of these models is to consider $k$-wise interactions, giving rise to generative models on random hypergraphs and tensors. Specifically, a hypergraphic version of the SBM was considered in \cite{ ghoshdastidar2014consistency,florescu2015spectral,angelini2015spectral}, and a version of the censored block model in \cite{ahn2016community}.

For the sake of exposition we restrict our attention to $k=4$ in the sequel. Most of our results however easily generalize to any $k$ and will be presented in a subsequent publication. In the remaining of this section, we will introduce a hypergraphic version of the SBM and its Gaussian analogue. 

\subsection{Hypergraphic SBM and its Gaussian analogue}

Here we describe a hypergraphic version of the stochastic block model. Let $y$ be a vector in $\{\pm 1\}^n$ such that $\one^T y = 0$. Let $p$ and $q$ be in $[0,1]$. Under the model, a random 4-uniform hypergraph $H$ on $n$ nodes is generated so that each $\{i,j,k,l\} \subseteq [n]$ is included in $E(H)$ independently with probability
\[
\begin{cases}
p & \text{ if $y_i=y_j=y_k=y_l$ (in the same community)} \\
q & \text{ otherwise (across communities).}
\end{cases}
\]
Let $\bA_H$, the adjacency tensor of $H$, be the 4-tensor given by
$$
(\bA_H)_{ijkl} = \begin{cases}
1 & \text{ if $\{i,j,k,l\} \in E(H)$} \\
0 & \text{ otherwise.}
\end{cases}
$$
Let $y^{\oeq 4}$ be the 4-tensor defined as
\[
y_{ijkl}^{\oeq 4} = \begin{cases}  1 & \text{ if } y_i=y_j=y_k=y_l \\ 0 &   \text{ otherwise,}\end{cases} 
\]
Since $y \in \{\pm 1\}^n$, we have
$$
y^{\oeq 4} =  \left(\frac{\one+y}{2}\right)^{\otimes 4} + \left(\frac{\one-y}{2}\right)^{\otimes 4}.
$$
Note that for any quadruple of distinct nodes $(i,j,k,l)$ we have
$$
(\EE \bA_H)_{ijkl} = (q\one^{\otimes 4} + (p-q)y^{\oeq 4})_{ijkl}.
$$

In the asymptotic regime of $p = a\log n/\binom{n-1}{3}$ and $q = b\log n /\binom{n-1}{3}$, there is a sharp information-theoretic threshold for exact recovery. Results regarding this hypergraphic stochastic block model will appear in subsequent publication. Here we will focus on the Gaussian counterpart.

Analogously to the relationship between the SBM and the spiked Wigner model, the hypergraphic version of the SBM suggests the following spiked tensor model:
\[
\bT = y^{\oeq 4} + \sigma \bW,
\]
where $\bW$ is a random 4-tensor with i.i.d. standard Gaussian entries. We note that here the noise tensor $\bW$ is not symmetric, unlike in the spiked Wigner model. This is not crucial: assuming $\bW$ to be a symmetric tensor will only scale $\sigma$ by $\sqrt{4!}$. 

\section{The Gaussian planted bisection model}

Given a sample $\bT$, our goal is to recover the hidden spike $y$ up to a global sign flip. Let $\widehat{y}$ be an estimator of $y$ computed from $\bT$. Let $p(\widehat{y};\sigma)$ be the probability that $\widehat{y}$ successfully recovers $y$. Since $\bW$ is Gaussian, $p(\widehat{y};\sigma)$ is maximized when
\[
\widehat{y} = \argmin_{x \in \{\pm 1\}^n: \one^T x = 0} \|\bT - x^{\oeq 4} \|_F^2,
\]
or equivalently $\widehat{y} = \argmax_x \ip{x^{\oeq 4}, \bT}$, the maximum-likelihood estimator (MLE) of $y$.

Let $f(x) = \ip{x^{\oeq 4}, \bT}$ and $\widehat{y}_{ML}$ be the MLE of $y$. By definition,
$$
p(\widehat{y}_{ML};\sigma) = \Pr \left(f(y) > \max_{\substack{x \in \{\pm 1\}\setminus \{y,-y\} \\ \one^T x = 0}} f(x) \right).
$$

\begin{theorem}
\label{thm:A}
Let $\epsilon > 0$ be a constant not depending on $n$. Then, as $n$ grows, $p(\widehat{y}_{ML};\sigma)$ converges to 1 if $\sigma < (1-\epsilon) \sigma^*$ and $p(\widehat{y}_{ML};\sigma)$ converges to 0 if $\sigma > (1+\epsilon) \sigma^*$, where
$$
\sigma^* = \sqrt{\frac{1}{8}} \cdot \frac{n^{3/2}}{\sqrt{\log n}}.
$$
\end{theorem}

Here we present a sketch of the proof while deferring the details to the appendix. Observe that $\widehat{y}_{ML}$ is not equal to $y$ if there exists $x \in \{\pm 1\}^n$ distinct from $y$ such that $\one^T x = 0$ and $f(x) \geq f(y)$. For each fixed $x$, the difference $f(x)-f(y)$ is equal to 
$$
\ip{y^{\oeq 4}, x^{\oeq 4} - y^{\oeq 4}} + \sigma \ip{\bW, x^{\oeq 4} - y^{\oeq 4}}
$$
which is a Gaussian random variable with mean $\ip{y^{\oeq 4}, x^{\oeq 4}-y^{\oeq 4}}$ and variance $\sigma^2 \|x^{\oeq 4}-y^{\oeq 4}\|_F^2$. By definition we have 
$$
\ip{x^{\oeq 4},y^{\oeq 4}} = \frac{1}{128}\left((\one^T \one + x^T y)^4 + (\one^T \one - x^T y)^4\right).
$$
Let $\phi(t) = \frac{1}{128}((1+t)^4+(1-t)^4)$. Then, $\ip{x^{\oeq 4}, y^{\oeq 4}} = n^4 \phi(x^T y/n)$ so
\begin{eqnarray*}
\ip{y^{\oeq 4}, x^{\oeq 4}-y^{\oeq 4}} &=& -n^4\left(\phi(1) - \phi(x^T y/n)\right), \\
\|x^{\oeq 4} - y^{\oeq 4}\|_F &=& n^2 \sqrt{2\phi(1)-2\phi(x^T y/n)}.
\end{eqnarray*}
Hence, $\Pr \left(f(x) - f(y) \geq 0\right)$ is equal to
\[
\Pr_{G \sim N(0,1)} \left(G \geq \frac{n^2}{\sigma \sqrt{2}} \cdot \sqrt{\phi(1)-\phi(x^T y/n)}\right).
\]
This probability is maximized when $x$ and $y$ differs by only two indices, that is, $x^T y = n - 4$. Indeed one can formally prove that the probability that $y$ maximizes $f(x)$ is dominated by the probability that $f(y) > f(x)$ for all $x$ with $x^T y = n-4$. By union bound and standard Gaussian tail bounds, the latter probability is at most
$$
\left(\frac{n}{2}\right)^2 \exp\left(-\frac{n^4}{4\sigma^2} \cdot \left(\phi(1)-\phi(1-4/n)\right)\right)
$$
which is approximately
$$
\exp\left(2\log n - \frac{n^3 \cdot \phi'(1)}{\sigma^2}\right) = \exp\left(2\log n - \frac{n^3}{4\sigma^2}\right),
$$
and it is $o_n(1)$ if $\sigma^2 < \frac{n^3}{8\log n}$ as in Theorem \ref{thm:A}. 

\section{Efficient recovery}

Before we address the Gaussian model, let us describe an algorithm for hypergraph partitioning. Let $H$ be a 4-uniform hypergraph on the vertex set $V = [n]$. Let $\bA_H$ be the adjacency tensor of $H$. Then the problem can be formulated as
\[
\max \ip{\bA_H, x^{\oeq 4}} \text{ subject to } x \in \{\pm 1\}^n, \one^T x = 0.
\]
One approach for finding a partition is to consider the multigraph realization of $H$, which appears in \cite{ghoshdastidar2014consistency,florescu2015spectral} in different terminology. Let $G$ be the multigraph on $V=[n]$ such that the multiplicity of edge $\{i,j\}$ is the number of hyperedges $e \in E(H)$ containing $\{i,j\}$. One may visualize it as substituting each hyperedge by a 4-clique. Now, one may attempt to solve the reduced problem
$$
\max \ip{A_G, xx^T} \text{ subject to } x\in \{\pm 1\}^n, \one^T x = 0
$$
which is unfortunately NP-hard in general. Instead, we consider the semidefinite programming (SDP) relaxation
\[
\begin{aligned}
& & \max \quad & \ip{A_G, X}\\
& & \text{ subject to} \quad & X_{ii}=1 \text{ for all $i \in [n]$}, \\
& & &  \ip{X, \one \one^T} = 0, \\
& & & X\succeq 0.
\end{aligned}
\]
When $A_G$ is generated under the graph stochastic block model, this algorithm recovers the hidden partition down to optimum parameters. On the other hand, it achieves recovery to nearly-optimum parameters when $G$ is the multigraph corresponding to the hypergraph generated by a hypergraphic stochastic block model: there is a constant multiplicative gap between the guarantee and the information-theoretic limit. This will be treated in a future publication.

Here we had two stages in the algorithm: (1) ``truncating'' the hypergraph down to a multigraph, and (2) relaxing the optimization problem on the truncated objective function.

Now let us return to the Gaussian model $\bT = y^{\oeq 4} + \sigma \bW$. Let $f(x) = \ip{x^{\oeq 4}, \bT}$. Our goal is to find the maximizer of $f(x)$. Note that 
\[
f(x) = \sum_{i_1,\cdots,i_4 \in [n]} \bT_{i_1 i_2 i_3 i_4} \cdot \frac{1}{16} \left(\prod_{s=1}^4 (1+x_{i_s}) + \prod_{s=1}^4 (1-x_{i_s})\right)
\]
so $f(x)$ is a polynomial of degree 4 in variables $x_1,\dotsc,x_n$. Let $f_{(2)}(x)$ be the degree 2 truncation of $f(x)$, i.e.,
$$
f_{(2)}(x) = \frac{1}{8} \sum_{i_1,\cdots,i_4 \in [n]} \bT_{i_1 i_2 i_3 i_4} \left(\sum_{1\leq s < t \leq 4} x_{i_s}x_{i_t}\right).
$$
Here we have ignored the constant term of $f(x)$ since it does not affect the maximizer. For each $\{s < t\}\subseteq \{1,2,3,4\}$, let $Q^{st}$ be $n$ by $n$ matrix where
$$
Q^{st}_{ij} = \sum_{\substack{(i_1,\cdots,i_4)\in [n]^4 \\ i_s = i, i_t = j}} \bT_{i_1 i_2 i_3 i_4}.
$$
Then,
$$
f_{(2)}(x) = \frac{1}{8} \ip{Q, xx^T}
$$
where $Q = \sum_{1\leq s<t\leq 4} Q^{st}$. This $Q$ is analogous to the adjacency matrix of the multigraph constructed above. It is now natural to consider the following SDP:
\begin{equation}
\label{eqn:sdp}
\begin{aligned}
&&\max \quad& \ip{Q, X}\\
&&\text{subject to}\quad& X_{ii}=1 \text{ for all $i\in [n]$},\\
&&& \ip{X, \one\one^T} = 0,\\
&&& X \succeq 0.
\end{aligned}
\end{equation}

\begin{theorem}
\label{thm:B}
Let $\epsilon > 0$ be a constant not depending on $n$. Let $\widehat{Y}$ be a solution of (\ref{eqn:sdp}) and $p(\widehat{Y};\sigma)$ be the probability that $\widehat{Y}$ coincide with $yy^T$. If $\sigma < (1-\epsilon) \sigma^*_{(2)}$ where
$$
\sigma^*_{(2)} = \sqrt{\frac{3}{32}} \cdot \frac{n^{3/2}}{\sqrt{\log n}} = \sqrt{\frac{3}{4}} \cdot \sigma^*,
$$
then $p(\widehat{Y};\sigma) = 1-o_n(1)$. 
\end{theorem}

We present a sketch of the proof where details are deferred to the appendix. We note that a similar idea was used in \cite{hajek2016achieving,Bandeira2016laplacian} for the graph stochastic block model.

We construct a dual solution of (\ref{eqn:sdp}) which is feasible with high probability, and certifies that $yy^T$ is the unique optimum solution for the primal (\ref{eqn:sdp}). By complementary slackness, such dual solution must be of the form $S := D_{Q'} - Q'$ where $Q' = \diag(y)Q\diag(y)$ and $D_{Q'}$ is $\mathrm{diag}(Q'\one)$. It remains to show that $S$ is positive semidefinite with high probability.

To show that $S$ is positive semidefinite, we claim that the second smallest eigenvalue of $\EE S$ is $\Theta(n^3)$ and the operator norm $\|S - \EE S\|$ is $O(\sigma n^{3/2} \sqrt{\log n})$ with high probability. The first part is just an easy calculation, and the second part is application of a nonasymptotic bound on Laplacian random matrices \cite{Bandeira2016laplacian}. Hence, $S$ is positive semidefinite with high probability if $\sigma n^{3/2} \sqrt{\log n} \lesssim n^3$, matching with the order of $\sigma_{(2)}^*$.

\section{Standard Spiked Tensor Model}

Montanari and Richard proposed a statistical model for tensor Principal Component Analysis \cite{Montanari2014}. In the model we observe a random tensor $\bT = v^{\otimes 4} + \sigma \bW$ where $v \in \RR^n$ is a vector with $\|v\| = \sqrt{n}$ (spike), $\bW$ is a random 4-tensor with i.i.d. standard Gaussian entries, and $\sigma \geq 0$ is the noise parameter. They showed a nearly-tight information-theoretic threshold for approximate recovery: when $\sigma \gg n^{3/2}$ then the recovery is information-theoretically impossible, while if $\sigma \ll n^{3/2}$ then the MLE gives a vector $v' \in \RR^n$ with $|v^T v'| = (1-o_n(1))n$ with high probability. Subsequently, sharp phase transitions for weak recovery, and strong and weak detection were shown in \cite{Perry2016,lesieur2017statistical}.

Those information-theoretic thresholds are achieved by the MLE for which no efficient algorithm is known. Montanari and Richard considered a simple spectral algorithm based on tensor unfolding, which is efficient in both theory and practice. They show that the algorithm finds a solution $v'$ with $|v^T v'| = (1-o_n(1))n$ with high probability as long as $\sigma = O(n)$ \cite{Montanari2014}. This is somewhat believed to be unimprovable using semidefinite programming \cite{Hopkins2015, Bhattiprolu2016}, or using approximate message passing algorithms \cite{lesieur2017statistical}. 

For clear comparison to the Gaussian planted bisection model, let us consider when spike is in the hypercube $\{\pm 1\}^n$. Let $y \in \{\pm 1\}^n$ and $\sigma \geq 0$. Given a tensor
$$
\bT = y^{\otimes 4} + \sigma \bW
$$
where $\bW$ is a random 4-tensor with independent, standard Gaussian entries, we would like to recover $y$ \emph{exactly}. The MLE is given by the maximizer of $g(x):= \ip{x^{\otimes 4}, \bT}$ over all vectors $x \in \{\pm 1\}^n$.

\begin{theorem}
Let $\epsilon > 0$ be any constant which does not depend on $n$. Let $\lambda^* = \frac{\sqrt{2} \cdot n^{3/2}}{\sqrt{\log n}}$. When $\sigma > (1+\epsilon)\lambda^*$, exact recovery is information-theoretically impossible (i.e. the MLE fails with $1-o_n(1)$ probability), while if $\sigma < (1-\epsilon)\lambda^*$ then the MLE recovers $y$ with $1-o_n(1)$ probability. 
\end{theorem}

The proof is just a slight modification of the proof of Theorem \ref{thm:A} which appears in the appendix. Note that both $\lambda^*$ and $\sigma^*$ are in the order of $n^{3/2}/\sqrt{\log n}$. The standard spiked tensor model and the Gaussian planted bisection model exhibit similar behaviour when unbounded computational resources are given.

\subsection{Sum-of-Squares based algorithms}

Here we briefly discuss Sum-of-Squares based relaxation algorithms. Given a polynomial $p \in \RR[x_1,\cdots,x_n]$, consider the problem of finding the maximum of $p(x)$ over $x \in \RR^n$ satisfying polynomial equalities $q_1(x)=0,\cdots,q_m(x)=0$. Most hard combinatorial optimization problems can be reduced into this form, including max-cut, $k$-colorability, and general constraint satisfaction problems. The \emph{Sum-of-Squares hierarchy} (SoS) is a systematic way to relax a polynomial optimization problem to a sequence of increasingly strong convex programs, each leading to a larger semidefinite program. See \cite{Barak2014} for a good exposition of the topic. 

There are many different ways to formulate the SoS hierarchy \cite{shor1987approach, parrilo2000structured,
nesterov2000squared, lasserre2001global}. Here we choose to follow the description based on \emph{pseudo-expectation functionals} \cite{Barak2014}. 

For illustration, let us recall the definition of $g(x)$: given a tensor $\bT = y^{\otimes 4} + \sigma \bW$, we define $g(x) = \ip{x^{\otimes 4}, \bT}$. Here $g(x)$ is a polynomial of degree 4, and the corresponding maximum-likelihood estimator is the maximizer
\begin{equation}
\label{eqn:polyopt}
\begin{aligned}
&&\max \quad& g(x) \\
&&\text{subject to}\quad& x_i^2 = 1 \text{ for all }i \in [n], \\
&&& \sum_{i=1}^n x_i = 0, x \in \RR^n.
\end{aligned}
\end{equation}
Let $\mu$ be a probability distribution over the set $\{x \in \{\pm 1\}^n : \one^T x = 0\}$. We can rewrite (\ref{eqn:polyopt}) as $\max_{\mu} \EE_{x\sim \mu} g(x)$ over such distributions. A linear functional $\widetilde{\EE}$ on $\RR[x]$ is called \emph{pseudoexpectation} of degree $2\ell$ if it satisfies $\widetilde{\EE}\, 1 = 1$ and $\widetilde{\EE}\, q(x)^2 \geq 0$ for any $q \in \RR[x]$ of degree at most $\ell$. We note that any expectation $\EE_{\mu}$ is a pseudoexpectation, but the converse is not true. So (\ref{eqn:polyopt}) can be relaxed to
\begin{equation}
\begin{aligned}
&&\max\quad& \widetilde{\EE}\, g(x) \\
&&\text{subject to}\quad& \widetilde{\EE}\text{ is a pseudoexpectation of degree $2\ell$,}\\
&&& \widetilde{\EE} \text{ is zero on $I$}
\end{aligned}
\end{equation}
where $I \subseteq \RR[x]$ is the ideal generated by $\sum_{i=1}^n x_i$ and $\{x_i^2 - 1\}_{i \in [n]}$. 

The space of pseudoexpectations is a convex set which can be described as an affine section of the semidefinite cone. As $\ell$ increases, this space gets smaller and in this particular case, it coincides with the set of true expectations when $\ell=n$. 

\subsection{SoS on spiked tensor models}

We would like to apply the SoS algorithm to the spiked tensor model. In particular, let us consider the degree 4 SoS relaxation of the maximum-likelihood problem. We note that for the spike tensor model with the spherical prior, it is known that neither algorithms using tensor unfolding \cite{Montanari2014}, the degree 4 SoS relaxation of the MLE \cite{Hopkins2015}, nor approximate message passing \cite{lesieur2017statistical} can achieve the statistical threshold, therefore the statistical-computational gap is believed to be present. Moreover, higher degree SoS relaxations were considered in \cite{Bhattiprolu2016}: for any small $\epsilon \geq 0$, degree at least $n^{\Omega(\epsilon)}$ is required in order to achieve recovery when $\sigma \approx n^{1+\epsilon}$.

We show an analogous gap of the degree 4 SoS relaxation for $\{\pm 1\}^n$ and $\one^T x=0$ prior.

\begin{theorem}
\label{thm:C}
Let $\bT = y^{\otimes 4} + \sigma \bW$ be as defined above. Let $g(x) = \ip{x^{\otimes 4}, \bT}$. If $\sigma \lesssim \frac{n}{\log^{\Theta(1)} n}$, then the degree 4 SoS relaxation of $\max_x g(x)$ gives the solution $y$, i.e., $\widetilde{\EE} g(x)$ is maximized when $\widetilde{\EE}$ is the expectation operator of the uniform distribution on $\{y, -y\}$.

On the other hand, if $\sigma \gtrsim n \log^{\Theta(1)} n$ then there exists a pseudoexpectation $\widetilde{\EE}$ of degree 4 on the hypercube $\{\pm 1\}^n$ satisfying $\one^T x = 0$ such that 
$$
g(y) < \max_{\widetilde{\EE}} \widetilde{\EE}\, g(x),
$$
so $y$ is not recovered via the degree 4 SoS relaxation.
\end{theorem}

The proof of Theorem \ref{thm:C} is very similar to one that appears in \cite{Hopkins2015, Bhattiprolu2016}. In the proof it is crucial to observe that
$$
\frac{n^3}{\log^{\Theta(1)}n}\lesssim \max_{\substack{\widetilde{\EE}:\text{degree 4}\\ \text{ pseudo-exp.}}} \widetilde{\EE}\, \ip{\bW, x^{\otimes 4}} \lesssim n^3 \log^{\Theta(1)} n.
$$
The upper bound can be shown via Cauchy-Schwarz inequality for pseudoexpectations and the lower bound, considerably more involved, can be shown by constructing a pseudoexpectation $\widetilde{\EE}$ which is highly correlated to the entries of $\bW$. We refer the readers to the appendix for details.

\subsection{Comparison with the planted bisection model}

Here we summarize the statistical-computational thresholds of two models, the planted bisection model and the spiked tensor model.
\begin{itemize}
\item For the planted bisection model, there is a constant $c > 0$ not depending on $n$ such that for any $\epsilon > 0$ if $\sigma < (c-\epsilon) \frac{n^{3/2}}{\sqrt{\log n}}$ then recovery is information-theoretically possible, and if $\sigma > (c+\epsilon) \frac{n^{3/2}}{\sqrt{\log n}}$ then the recovery is impossible via any algorithm. Moreover, there is an efficient algorithm and a constant $c' < c$ such that the algorithm recovers $y$ with high probability when $\sigma < c' \frac{n^{3/2}}{\sqrt{\log n}}$. 

\item For the spiked tensor model, there is a constant $C > 0$ not depending on $n$ such that for any $\epsilon > 0$ if $\sigma < (C-\epsilon) \frac{n^{3/2}}{\sqrt{\log n}}$ then the recovery is information-theoretically possible, and if $\sigma > (C+\epsilon) \frac{n^{3/2}}{\sqrt{\log n}}$ then the recovery is impossible. In contrast to the planted bisection model, all efficient algorithms known so far only successfully recover $y$ in the asymptotic regime of $\sigma \lesssim n$. 
\end{itemize}

We note that the efficient, nearly-optimal recovery is achieved for the planted bisection model by ``forgetting'' higher moments of the data. On the other hand, such approach is unsuitable for the spiked tensor model since the signal $y^{\otimes 4}$ does not seem to be approximable by a non-trivial low-degree polynomial. This might shed light into an interesting phenomenon in average complexity, as in this case it seems to be crucial that second moments (or pairwise relations) carry information.

We further notice that the difference between the spiked tensor model and the tensor model with planted non-overlapping clusters in the \emph{detection} problem was studied in \cite{lesieur2017statistical}. They prove that when the mean of the prior is zero (which corresponds to the spiked tensor model), then there is a huge gap between the information-theoretic bound and approximate message passing algorithm, while when the mean is non-zero (which corresponds to the planted bisection model in this paper) the detection is possible in the same regime as the information-theoretic threshold. It leads us to ask whether this computational versus information-theoretic gap can be classified in terms of the ``complexity'' of the signal which generalizes (i) the minimum degree of nonzero Fourier coefficient in exact recovery, and (ii) the mean of the prior in detection.

All information-theoretic phase transition exhibited in the paper readily generalize to $k$-tensor models and this will be discussed in a future publication. In future work we will also investigate the performance of higher degree SoS algorithms for the planted bisection model.

\bibliographystyle{alpha}
\bibliography{mybib}

\appendix

\section{Proof of Theorem \ref{thm:A}}
\label{app:A}

In this section, we prove Theorem \ref{thm:A} for general $k$.

\begin{theorem}[Theorem \ref{thm:A}, general $k$]
\label{thm:A2}
Let $k$ be a positive integer with $k \geq 2$. Let $y \in \{\pm 1\}^n$ with $\one^T y = 0$ and $\bT$ be a $k$-tensor defined as $\bT = y^{\oeq k} + \sigma \bW$ where $\bW$ is a random $k$-tensor with i.i.d. standard Gaussian entries. Let $\widehat{y}_{ML}$ be the maximum-likelihood estimator of $y$, i.e.,
$$
\widehat{y}_{ML} = \argmax_{x \in \{\pm 1\}^n: \one^T x} \ip{\bT, x^{\oeq k}}.
$$
For any positive $\epsilon$,
\begin{itemize}
\item[(i)] $\widehat{y}_{ML}$ is equal to $y$ with probability $1-o_n(1)$ if $\sigma < (1-\epsilon)\sigma_*$, and
\item[(ii)] $\widehat{y}_{ML}$ is not equal to $y$ with probability $1-o_n(1)$ if $\sigma > (1+\epsilon)\sigma_*$
\end{itemize}
where 
$$
\sigma_* = \sqrt{\frac{k}{2^k}} \cdot \frac{n^{\frac{k-1}{2}}}{\sqrt{2\log n}}.
$$
\end{theorem}

First we prove the following lemma.

\begin{lemma}
Let $\phi$ be a function defined as $\phi(t) = \frac{1}{2^{2k-1}} \left((1-t)^k + (1+t)^k\right)$. Then, 
$$
\ip{x^{\oeq k}, y^{\oeq k}} = n^k \phi\left(\frac{x^T y}{n}\right)
$$
for any $x, y\in \{\pm 1\}^n$ such that $\one^T x= \one^T y=0$.
\end{lemma}

\begin{proof}
Note that $x^{\oeq k} = \frac{1}{2^k} \left((\one+x)^{\otimes k} + (\one-x)^{\otimes k}\right)$. Hence,
$$
\ip{x^{\oeq k},y^{\oeq k}} = \frac{1}{2^{2k}}\sum_{s,t \in \{\pm 1\}} \ip{\one+sx, \one+ty}^k.
$$
Since $\one^T x = \one^T y = 0$, we have
\begin{eqnarray*}
\ip{x^{\oeq k},y^{\oeq k}} &=& \frac{1}{2^{2k}} \left(2(\one^T \one + x^T y)^k + 2(\one^T \one - x^T y)^k\right) \\
&=& \frac{n^k}{2^{2k-1}} \left(\left(1+\frac{x^T y}{n}\right)^k + \left(1-\frac{x^T y}{n}\right)^k\right)\\
&=& n^k \phi\left(\frac{x^T y}{n}\right)
\end{eqnarray*}
as desired. 
\end{proof}

\begin{proof}[Proof of Theorem \ref{thm:A2}]

Let
$$
f(x) = \ip{x^{\oeq k}, \bT} = \ip{x^{\oeq k}, y^{\oeq k}+\sigma \bW}.
$$
By definition, $\widehat{y}_{ML}$ is not equal to $y$ if there exists $x \in \{\pm 1\}^n$ distinct from $y$ or $-y$ such that $\one^T x = 0$ and $f(x)$ is greater than or equal to $f(y)$. For each fixed $x \in \{\pm 1\}^n$ with $\one^T x=0$, note that 
$$
f(x)-f(y) = \ip{y^{\oeq k}, x^{\oeq k}-y^{\oeq k}} + \sigma \ip{\bW, x^{\oeq k}-y^{\oeq k}}
$$
is a Gaussian random variable with mean $-n^k(\phi(1)-\phi(x^T y/n)$ and variance
$$
\sigma^2 \|x^{\oeq k}-y^{\oeq k}\|_F^2 = 2\sigma^2 n^k (\phi(1)-\phi(x^T y/n)).
$$
Hence, $\Pr \left(f(x) - f(y) \geq 0\right)$ is equal to
\[
\Pr_{G \sim N(0,1)} \left(G \geq \frac{n^{k/2}}{\sigma \sqrt{2}} \cdot \sqrt{\phi(1)-\phi\left(\frac{x^T y}{n}\right)}\right).
\]

Let $\sigma_*$ be
$$
\sigma_* = \sqrt{\phi'(1)} \cdot \frac{n^{\frac{k-1}{2}}}{\sqrt{2\log n}}.
$$
Since $\phi'(1) = \frac{k}{2^k}$, it matches with the definition in the statement of the Theorem.

\noindent \emph{Upper bound.} Let us prove that $p(\widehat{y}_{ML};\sigma) = 1-o_n(1)$ if $\sigma < (1-\epsilon)\sigma_*$.
By union bound, we have
\begin{eqnarray*}
1-p(\widehat{y}_{ML};\sigma) &\leq & \sum_{\substack{x \in \{\pm 1\}^n\setminus \{y,-y\} \\ \one^T x = 0}} \Pr \left(f(x) - f(y) \geq 0 \right) \\
&=& \sum_{\substack{x \in \{\pm 1\}^n \setminus \{y,-y\} \\ \one^T x = 0}} \Pr_{G \sim N(0,1)} \left(G \geq \frac{n^{k/2}}{\sigma \sqrt{2}} \cdot \sqrt{\phi(1) - \phi\left(\frac{x^T y}{n}\right)}\right) \\
&\leq & \sum_{\substack{x \in \{\pm 1\}^n \setminus \{y,-y\} \\ \one^T x = 0}} \exp\left(-\frac{n^k}{4\sigma^2} \left(\phi(1) - \phi\left(\frac{x^Ty}{n}\right)\right)\right).
\end{eqnarray*}
The last inequality follows from a standard Gaussian tail bound $\Pr_G (G > t) \leq \exp(-t^2/2)$.

A simple counting argument shows that the number of $x \in \{\pm 1\}^n$ with $\one^T x = 0$ and $x^T y = n-4r$ is exactly $\binom{n/2}{r}^2$ for $r \in \{0,1,\cdots,n/2\}$. Moreover, for any $t \geq 0$ we have $\phi(t)=\phi(-t)$. Hence,
$$
1 - p(\widehat{y}_{ML};\sigma) \leq 2\sum_{r=1}^{\ceil{n/4}} \binom{n/2}{r}^2 \exp\left(-\frac{n^k}{4\sigma^2} \left(\phi(1) - \phi\left(1-\frac{4r}{n}\right)\right)\right).
$$

Note that $\phi(1)-\phi(1-x) \geq \phi'(1) x - O(x^2)$. Hence, there exists an absolute constant $C > 0$ such that $\phi(1)-\phi(1-4r/n)$ is at least $(1-\epsilon)\phi'(1)\cdot 4r/n$ if $r < Cn$ and is at least $\Omega(1)$ otherwise. Since $\sigma < (1-\epsilon)\sigma^*$ we have
$$
-\frac{n^k}{4\sigma^2} \left(\phi(1)-\phi\left(1-\frac{4r}{n}\right)\right) \leq -\frac{2r\log n}{1-\epsilon} \leq -2(1+\epsilon)r\log n
$$
if $r < Cn$, and $-\frac{n^k}{4\sigma^2} \left(\phi(1)-\phi\left(1-\frac{4r}{n}\right)\right) = -\Omega(n\log n)$ otherwise. It implies that
\begin{eqnarray*}
1-p(\widehat{y}_{ML};\sigma) &\leq& 2 \left(\sum_{r=1}^{Cn}\exp(-2\epsilon r\log n) + n\exp(-\Omega(n\log n))\right) \\
&\lesssim & n^{-2\epsilon} + n\exp(-\Omega(n\log n)) = o_n(1).
\end{eqnarray*}

\noindent\emph{Lower bound.} Now we prove that $p(\widehat{y}_{ML};\sigma) = o_n(1)$ if $\sigma > (1-\epsilon)\sigma^*$.
Let $A = \{i \in [n] : y_i = +1\}$ and $B = [n]\setminus A$. For each $a \in A$ and $b \in B$, let $E_{ab}$ be the event that $f(y^{(ab)})$ is greater than $f(y)$ where $y^{(ab)}$ is the $\pm 1$-vector obtained by flipping the signs of $y_a$ and $y_b$. For any $H_A \subseteq A$ and $H_B \subseteq B$, note that
\begin{eqnarray*}
1-p(\widehat{y}_{ML};\sigma) &\geq& \Pr\left(\bigcup_{a \in H_A, b\in H_B} E_{ab}\right) \\
&=& \Pr\left(\max_{a\in H_A, b \in H_B} \left(f(y^{(ab)})-f(y)\right) > 0\right)
\end{eqnarray*}
since any of the event $E_{ab}$ implies that $\widehat{y}_{ML} \neq y$. Recall that 
$$
f(y^{(ab)}) - f(y) = -n^k\left(\phi(1)-\phi\left(1-\frac{4}{n}\right)\right) + \sigma \ip{\bW, (y^{(ab)})^{\oeq k} - y^{\oeq k}}.
$$
So, $1-p(\widehat{y}_{ML};\sigma)$ is at least
$$
\Pr\left(\max_{\substack{a\in H_A \\ b\in H_B}} \ip{\bW,(y^{(ab)})^{\oeq k}-y^{\oeq k}} > \frac{n^k}{\sigma} \left(\phi(1)-\phi\left(1-\frac4n\right)\right)\right).
$$

Fix $H_A \subseteq A$ and $H_B \subseteq B$ with $|H_A|=|H_B|=h$ where $h = \frac{n}{\log^2 n}$. Let $(\mathcal{X},\mathcal{Y},\mathcal{Z})$ be the partition of $[n]^k$ defined as
\begin{eqnarray*}
\mathcal{X} &=& \{ \alpha \in [n]^k : \alpha^{-1}(H_A \cup H_B)=\emptyset \}, \\
\mathcal{Y} &=& \{ \alpha \in [n]^k : |\alpha^{-1}(H_A \cup H_B)|=1 \}, \\
\mathcal{Z} &=& \{ \alpha \in [n]^k : |\alpha^{-1}(H_A \cup H_B)|\geq 2 \}.
\end{eqnarray*}
Let $\bW_{\mathcal{X}}$, $\bW_{\mathcal{Y}}$ and $\bW_{\mathcal{Z}}$ be the $k$-tensor supported on $\mathcal{X}$ , $\mathcal{Y}$, $\mathcal{Z}$ respectively. For each $a\in H_A$ and $b\in H_B$, let
\begin{eqnarray*}
X_{ab} &=& \ip{\bW_\mathcal{X}, (y^{(ab)})^{\oeq k}-y^{\oeq k}},\\
Y_{ab} &=& \ip{\bW_\mathcal{Y}, (y^{(ab)})^{\oeq k}-y^{\oeq k}},\\
Z_{ab} &=& \ip{\bW_\mathcal{Z}, (y^{(ab)})^{\oeq k}-y^{\oeq k}}.\\
\end{eqnarray*}

\begin{claim}
The followings are true:
\begin{itemize}
\item[(i)] $X_{ab} = 0$ for any $a \in H_A$ and $b \in H_B$.
\item[(ii)] For fixed $a \in H_A$ and $b \in H_B$, the variables $Y_{ab}$ and $Z_{ab}$ are independent.
\item[(iii)] Each $Y_{ab}$ can be decomposed into $Y_a + Y_b$ where $\{Y_a\}_{a \in H_A} \cup \{Y_b\}_{b \in H_B}$ is a collection of i.i.d. Gaussian random variables.
\end{itemize}
\end{claim}

\begin{proof}[Proof of Claim]
First note that $\left((y^{(ab)})^{\oeq k}-y^{\oeq k}\right)_\alpha$ is non-zero if and only if $|\alpha^{-1}(a)|$ and $|\alpha^{-1}(b)|$ have the same parity. This implies (i) since when $\alpha \in \mathcal{X}$ we have $|\alpha^{-1}(a)| = |\alpha^{-1}(b)|=0$. (ii) holds because $\mathcal{Y} \cap \mathcal{Z} = 0$. 

For $s \in H_A \cup H_B$, let $\mathcal{Y}_{s}$ be the subset of $\mathcal{Y}$ such that 
$$
\mathcal{Y}_{s} = \{ \alpha \in \mathcal{Y}: |\alpha^{-1}(s)|=1 \}.
$$
By definition, $\mathcal{Y}_{s}$ are disjoint and $\mathcal{Y} = \bigcup_{s \in H_A \cup H_B} \mathcal{Y}_s$. Hence, 
$$
Y_{ab} = \sum_{s \in H_A \cup H_B} \ip{\bW_{\mathcal{Y}_s}, (y^{(ab)})^{\oeq k}-y^{\oeq k}}.
$$
Moreover, $\ip{\bW_{\mathcal{Y}_s}, (y^{(ab)})^{\oeq k}-y^{\oeq k}}$ is zero when $s \not\in \{a, b\}$. So,
$$
Y_{ab} = \sum_{\alpha \in \mathcal{Y}_a \cup \mathcal{Y}_b} \bW_\alpha ((y^{(ab)})^{\oeq k}-y^{\oeq k})_\alpha.
$$
Note that for $\alpha \in \mathcal{Y}_a$
$$
((y^{(ab)})^{\oeq k}-y^{\oeq k})_\alpha = \begin{cases}
+1 & \text{if $|\alpha^{-1}(A\setminus H_A)| = 0$} \\
-1 & \text{if $|\alpha^{-1}(B\setminus H_B)| = 0$} \\
0 & \text{otherwise}.
\end{cases}
$$
So,
$$
\sum_{\alpha \in \mathcal{Y}_a} \bW_\alpha ((y^{(ab)})^{\oeq k}-y^{\oeq k})_\alpha = \sum_{\substack{\alpha \in \mathcal{Y}_a \\ |\alpha^{-1}(A\setminus H_A)|=0}} \bW_{\alpha} - \sum_{\substack{\alpha \in \mathcal{Y}_a \\ |\alpha^{-1}(B\setminus H_B)|=0}} \bW_{\alpha}
$$
which does not depend on the choice of $b$. Let
$$
Y_a = \sum_{\substack{\alpha \in \mathcal{Y}_a \\ |\alpha^{-1}(A\setminus H_A)|=0}} \bW_{\alpha} - \sum_{\substack{\alpha \in \mathcal{Y}_a \\ |\alpha^{-1}(B\setminus H_B)|=0}} \bW_{\alpha}
$$
and $Y_b$ respectively. Then $Y_{ab} = Y_a + Y_b$ and $\{Y_s\}_{s \in H_A \cup H_B}$ is a collection of independent Gaussian random variables. Moreover, the variance of $Y_s$ is equal to $2k\left(\frac{n}{2}-h\right)^{k-1}$,
independent of the choice of $s$.
\end{proof}

By the claim, we have $\ip{\bW, (y^{(ab)})^{\oeq k}-y^{\oeq k}} = Y_a + Y_b + Z_{ab}$. Moreover,
\begin{eqnarray*}
\max_{\substack{a\in H_A\\b\in H_B}} (Y_a+Y_b+Z_{ab}) &\geq & 
\max_{\substack{a\in H_A\\b\in H_B}} (Y_a+Y_b) - \max_{\substack{a\in H_A\\b\in H_B}} (-Z_{ab})\\
&=& \max_{a \in H_A} Y_a + \max_{b \in H_B} Y_b - \max_{\substack{a\in H_A\\b\in H_B}} (-Z_{ab}).
\end{eqnarray*}

We need the following tail bound on the maximum of Gaussian random variables.
\begin{lemma}
Let $G_1,\dotsc,G_N$ be (not necessarily independent) Gaussian random variables with variance $1$. Let $\epsilon > 0$ be a constant which does not depend on $N$. Then,
$$
\Pr\left(\max_{i=1,\cdots,N} G_i > (1+\epsilon)\sqrt{2\log N} \right) \leq N^{-\epsilon}.
$$
On the other hand,
$$
\Pr\left(\max_{i=1,\cdots,N} G_i < (1-\epsilon)\sqrt{2\log N}\right) \leq \exp(-N^{\Omega(\epsilon)})
$$
if $G_i$'s are independent.
\end{lemma}

By the Lemma, we have
\begin{eqnarray*}
\max_{a \in H_A} Y_a &\geq& (1-0.01\epsilon)\sqrt{2\log h \cdot 2k \left(\frac{n}{2}-h\right)^{k-1}} \\
\max_{b \in H_A} Y_b &\geq& (1-0.01\epsilon)\sqrt{2\log h \cdot 2k \left(\frac{n}{2}-h\right)^{k-1}} \\
\max_{\substack{a \in H_A \\ b \in H_B}} Z_{ab} &\lesssim& \sqrt{\log h \cdot \max \var(Z_{ab})}
\end{eqnarray*}
with probability $1-o_n(1)$. Note that
\begin{eqnarray*}
\var(Z_{ab}) &=& \|(y^{(ab)})^{\oeq k}-y^{\oeq k}\|_F^2 - (\var(Y_A)+\var(Y_B)) \\
&=& 2n^k \left(\phi(1)-\phi\left(1-\frac{4}{n}\right)\right) - 4k \left(\frac{n}{2}-h\right)^{k-1} \\
&\leq & 8n^{k-1}\left(\phi'(1) - (1-o(1)) \frac{k}{2^k}\right)
\end{eqnarray*}
which is $o(n^{k-1})$. Hence,
\begin{eqnarray*}
\max_{\substack{a \in H_A \\ b \in H_B}} \ip{\bW, (y^{(ab)})^{\oeq k}-y^{\oeq k}} &\geq& 2(1-0.01\epsilon-o(1))\sqrt{2\log n \cdot 2k\left(\frac{n}{2}-h\right)^{k-1}} \\
&\geq & (1-0.01\epsilon-o(1))\sqrt{\frac{kn^{k-1}\log n}{2^{k-5}}}.
\end{eqnarray*}
On the other hand, since $\sigma > (1+\epsilon) \sigma^*$ we have
\begin{eqnarray*}
\frac{n^k}{\sigma}\left(\phi(1)-\phi\left(1-\frac{4}{n}\right)\right) &<& \frac{n^k}{1+\epsilon}\cdot \frac{4\phi'(1)}{n} \cdot \sqrt{\frac{2\log n}{n^{k-1}\phi'(1)}} \\
&<& \frac{1}{1+\epsilon} \sqrt{\frac{n^{k-1}\log n}{2^{k-5}}}
\end{eqnarray*}
which is less than
$$
\max_{\substack{a \in H_A \\ b \in H_B}} \ip{\bW, (y^{(ab)})^{\oeq k}-y^{\oeq k}} 
$$
with probability $1-o_n(1)$. Thus, $1-p(\widehat{y}_{ML})\geq 1-o_n(1)$. 
\end{proof}

\section{Proof of Theorem \ref{thm:B}}

In this section, we prove Theorem \ref{thm:B} for general $k$.

Let $k$ be a positive integer with $k > 2$. Let $y \in \{\pm 1\}^n$ with $\one^T y = 0$ and $\bT$ be a $k$-tensor defined as $\bT = y^{\oeq k} + \sigma \bW$ where $\bW$ is a random $k$-tensor with i.i.d. standard Gaussian entries. Let $f(x) = \ip{x^{\oeq k}, \bT}$ and let $f_{(2)}(x)$ be the degree 2 truncation of $f(x)$, i.e.,
$$
f_{(2)}(x) = \sum_{\alpha \in [n]^k} \bT_{\alpha} \left(\frac{1}{2^{k-1}} \sum_{1\leq s < t \leq k} x_{\alpha(s)}x_{\alpha(t)}\right).
$$
For each $\{s < t\}\subseteq [k]$, let $Q^{st}$ be $n$ by $n$ matrix where
$$
Q^{st}_{ij} = \frac{1}{2} \left(\sum_{\substack{\alpha \in [n]^k \\ \alpha(s) = i, \alpha(t) = j}} \bT_\alpha + \sum_{\substack{\alpha \in [n]^k \\ \alpha(s) = j, \alpha(t) = i}} \bT_\alpha\right)
$$
Then,
$$
f_{(2)}(x) = \frac{1}{2^{k-1}} \ip{Q, xx^T}
$$
where $Q = \sum_{1\leq s<t\leq k} Q^{st}$. We consider the following SDP relaxation for $\max_x f_{(2)}(x)$:
\begin{equation}
\label{eqn:sdpk}
\begin{aligned}
&&\max \quad& \ip{Q, X}\\
&&\text{subject to}\quad& X_{ii}=1 \text{ for all $i\in [n]$},\\
&&& \ip{X, \one\one^T} = 0,\\
&&& X \succeq 0.
\end{aligned}
\end{equation}

\begin{theorem}[Theorem \ref{thm:B}, general $k$]
\label{thm:B2}
Let $\epsilon > 0$ be a constant not depending on $n$. Let $\widehat{Y}$ be a solution of (\ref{eqn:sdpk}) and $p(\widehat{Y};\sigma)$ be the probability that $\widehat{Y}$ coincide with $yy^T$. Let $\sigma_{(2)}$ be 
$$
\sigma_{(2)} = \sqrt{\frac{k(k-1)}{2^{2k-1}}} \cdot \frac{n^{\frac{k-1}{2}}}{\sqrt{2\log n}}
$$
If $\sigma < (1-\epsilon) \sigma_{(2)}$, then $p(\widehat{Y};\sigma) = 1-o_n(1)$. 
\end{theorem}

\begin{proof}
First note that $Q^{st} = \frac{n^{k-2}}{2^{k-1}} yy^T + \sigma W^{st}$ where
$$
W^{st}_{ij} = \sum_{\substack{\alpha \in [n]^k \\ \alpha(s) = i, \alpha(t) = j}} \bW_{\alpha}.
$$
So we have
$$
Q = n^{k-2} \cdot \frac{k(k-1)}{2^k} yy^T + \sigma \widebar{W}
$$
where $\widebar{W} = \sum_{1\leq s < t \leq k} W^{st}$. 

The dual program of (\ref{eqn:sdpk}) is 
\begin{equation}
\label{eqn:sdpdual}
\begin{aligned}
&&\max \quad& \tr(D) \\
&&\text{subject to}\quad& D + \lambda \one\one^T - Q \succeq 0,\\
&&& D \text{ is diagonal}.
\end{aligned}
\end{equation}

By complementary slackness, $yy^T$ is the unique optimum solution of (\ref{eqn:sdpk}) if there exists a dual feasible solution $(D, \lambda)$ such that
$$
\ip{D+\lambda\one\one^T-Q, yy^T} = 0
$$
and the second smallest eigenvalue of $D+\lambda\one\one^T-Q$ is greater than zero. For brevity, let $S = D+\lambda \one\one^T - Q$. Since $S$ is positive semidefinite, we must have $Sy = 0$, that is,
$$
D_{ii} = \sum_{j = 1}^n Q_{ij} y_i y_j 
$$
for any $i \in [n]$. 

For a symmetric matrix $M$, we define the \emph{Laplacian} $\mathcal{L}(M)$ of $M$ as $\mathcal{L}(M) = \diag(M\one) - M$ (See \cite{Bandeira2016laplacian}). Using this language, we can express $S$ as
$$
S = \diag(y) \left(\mathcal{L}(\diag(y)Q\diag(y))+\lambda yy^T\right) \diag(y).
$$
Note that
$$
\diag(y)Q\diag(y) = n^{k-2} \cdot \frac{k(k-1)}{2^k} \one\one^T + \sigma \diag(y)\widebar{W}\diag(y).
$$
Hence, the Laplacian of $\diag(y)Q\diag(y)$ is equal to
$$
\underbrace{
n^{k-1} \cdot \frac{k(k-1)}{2^k} \left(I_{n\times n} - \frac{1}{n}\one\one^T\right)}_{\text{deterministic part}} + \underbrace{\sigma \mathcal{L}\left(\diag(y)\widebar{W}\diag(y)\right)}_{\text{noisy part}}.
$$
This matrix is positive semidefinite if
$$
\sigma \left\| \mathcal{L}(\diag(y)\widebar{W}\diag(y)\right\| \leq n^{k-1} \cdot \frac{k(k-1)}{2^k}.
$$
Moreover, if the inequality is strict then the second smallest eigenvalue of $S$ is greater than zero. 

By triangle inequality, we have
\begin{eqnarray*}
\left\|\mathcal{L}(\diag(y)\widebar{W} \diag(y))\right\| &\leq&  \max_{i \in [n]} \sum_{j=1}^n \widebar{W}_{ij} y_i y_j + \sum_{1\leq s<t\leq k} \|\diag(y)W^{st}\diag(y)\| \\
&\leq& \max_{i \in [n]} \sum_{j=1}^n \widebar{W}_{ij} y_i y_j + \binom{k}{2} \sqrt{2n^{k-1}}.
\end{eqnarray*}
The second inequality holds with high probability since $W^{st}$ has independent Gaussian entries. Since $\sum_{j=1}^n \widebar{W}_{ij} y_i y_j$ is Gaussian and centered, with high probability we have that
\begin{eqnarray*}
\max_{i \in [n]} \sum_{j=1}^n \widebar{W}_{ij} y_i y_j &\leq& (1+0.1\epsilon) \sqrt{2\log n} \cdot \left(\max_{i\in [n]} \var\left(\sum_{j=1}^n \widebar{W}_{ij} y_i y_j\right)\right)^{1/2} \\
&\leq& (1+0.1\epsilon) \sqrt{2\log n} \cdot \sqrt{\binom{k}{2}n^{k-1}}.
\end{eqnarray*}
Hence, 
$$
\left\|\mathcal{L}(\diag(y)\widebar{W} \diag(y))\right\| \leq (1+0.1\epsilon+o(1)) \sqrt{2\binom{k}{2} n^{k-1}\log n}
$$
with high probability. So, $S \succeq 0$ as long as
$$
\sigma (1+0.1\epsilon+o(1)) \sqrt{2\binom{k}{2} n^{k-1}\log n} < n^{k-1} \cdot \frac{k(k-1)}{2^k}
$$
or simply 
$$
(1+0.1\epsilon+o(1))\sigma < \sqrt{\frac{k(k-1)}{2^{2k-1}}} \cdot \frac{n^{\frac{k-1}{2}}}{\sqrt{2\log n}} = \sigma_{(2)}.
$$
So, when $\sigma < (1-\epsilon)\sigma_{(2)}$ with high probability $\widehat{Y} = yy^T$. 
\end{proof}

\section{Pseudo-expectation and its moment matrix}

\subsection{Notation and preliminaries}

Let $\cV$ be the space of real-valued functions on the $n$-dimensional hypercube $\{-1,+1\}^n$. For each $S \subseteq [n]$, let $x_S = \prod_{i \in S} x_i$. Note that $\{x_S : S\subseteq [n]\}$ is a basis of $\cV$. Hence, any function $f : \{\pm 1\}^n \to \RR$ can be written as a unique linear combination of multilinear monomials, say
$$
f(x) = \sum_{S \subseteq [n]} f_S x_S.
$$
The degree of $f$ is defined as the maximum size of $S \subset [n]$ such that $f_S$ is nonzero.

Let $\ell$ be a positive integer. Let us denote the collection of subsets of $[n]$ of size at most $\ell$ by $\binom{[n]}{\leq \ell}$, and the size $\left|\binom{[n]}{\leq \ell}\right|$ of it by $\binom{n}{\leq \ell}$. 

Let $M$ be a square symmetric matrix of size $\binom{n}{\leq \ell}$. The rows and the columns of $M$ are indexed by the elements in $\binom{[n]}{\leq \ell}$. To avoid confusion, we use $M[S,T]$ to denote the entry of $M$ at row $S$ and column $T$. 

We say $M$ is \emph{SoS-symmetric} if $M[S,T] = M[S',T']$ whenever $x_S x_T = x_{S'} x_{T'}$ on the hypercube. Since $x \in \{\pm 1\}^n$, it means that $S\oplus T = S' \oplus T'$ where $S \oplus T$ denotes the symmetric difference of $S$ and $T$. Given $f \in \cV$ with degree at most $2\ell$, we say $M$ \emph{represents} $f$ if
$$
f_U = \sum_{\substack{S,T \in \binom{[n]}{\leq \ell}:\\S\oplus T = U}} M[S,T].
$$
We use $M_f$ to denote the unique SoS-symmetric matrix representing $f$.

Let $\cL$ be a linear functional on $\cV$. By linearity, $\cL$ is determined by $(\cL[x_S]: S \subseteq [n])$. Let $X_\cL$ be the SoS-symmetric matrix of size $\binom{n}{\leq \ell}$ with entries $X[S,T] = \cL[x_{S\oplus T}]$. We call $X_\cL$ the \emph{moment matrix of $\cL$ of degree $2\ell$}. 

By definition, we have
\begin{eqnarray*}
\cL[f] &=& \sum_{U \in \binom{[n]}{\leq 2\ell}} f_U \cL[x_U] \\
&=& \sum_{S,T \in \binom{[n]}{\leq \ell}} M_f[S,T] X_\cL[S,T] \\
&=& \ip{X_\cL, M_f}
\end{eqnarray*}
for any $f$ of degree at most $2\ell$. 

\subsection{A quick introduction to pseudo-expectations}

For our purpose, we only work on pseudo-expectations defined on the hypercube $\{\pm 1\}^n$. See \cite{Barak2014} for general definition. 

Let $\ell$ be a positive integer and $d = 2\ell$. A \emph{pseudo-expectation of degree $d$ on $\{\pm 1\}^n$} is a linear functional $\Etilde$ on the space $\cV$ of functions on the hypercube such that
\begin{itemize}
\item[(i)] $\Etilde[1] = 1$,
\item[(ii)] $\Etilde[q^2] \geq 0$ for any $q \in \cV$ of degree at most $\ell = d/2$.
\end{itemize}
We say $\Etilde$ \emph{satisfies} the system of equalities $\{p_i(x) = 0\}_{i=1}^m$ if $\Etilde[f] = 0$ for any $f \in \cV$ of degree at most $d$ which can be written as
$$
f = p_1 q_1 + p_2 q_2 + \cdots + p_m q_m
$$
for some $q_1,q_2,\dotsc,q_m \in \cV$.

We note the following facts:
\begin{itemize}
\item If $\Etilde$ is a pseudo-expectation of degree $d$, then it is also a pseudo-expectation of any degree smaller than $d$. 
\item If $\Etilde$ is a pseudo-expectation of degree $2n$, then $\Etilde$ defines a valid probability distribution supported on $P:=\{x \in \{\pm 1\}^n: p_i(x) = 0 \text{ for all }i\in[m]\}$. 
\end{itemize}

The second fact implies that maximizing $f(x)$ on $P$ is equivalent to maximizing $\Etilde[f]$ over all pseudo-expectations of degree $2n$ satisfying $\{p_i(x)=0\}_{i=1}^m$.
Now, let $d$ be an even integer such that 
$$
d > \max\{\deg(f),\deg(p_1),\cdots,\deg(p_m)\}.
$$
We relax the original problem to
\begin{equation}
\tag{$\mathsf{SoS}_d$}
\label{eqn:polysos}
\begin{aligned}
& & \max \quad & \Etilde[f]  \\
& & \text{ subject to }\quad & \Etilde \text{ is degree-$d$ pseudo-expectation on $\{\pm 1\}^n$} \\
& & & \text{ satisfying $\{p_i(x)=0\}_{i=1}^m$}.
\end{aligned}
\end{equation}

We note that the value of (\ref{eqn:polysos}) decreases as $d$ grows, and it reaches the optimum value $\max_{x\in P} p_0(x)$ of the original problem at $d = 2n$. 

\subsection{Matrix point of view}

Let $\Etilde$ be a pseudo-expectation of degree $d = 2\ell$ for some positive integer $\ell$. Suppose that $\Etilde$ satisfies the system $\{p_i(x) = 0\}_{i=1}^m$. Let $X_{\Etilde}$ be the moment matrix of $\Etilde$ of degree $2\ell$, hence the size of $X_{\Etilde}$ is $\binom{n}{\leq \ell}$. 

Conditions for $\Etilde$ being pseudo-expectation translate as the following conditions for $X_{\Etilde}$: 
\begin{itemize}
\item[(i)] $X_{\emptyset,\emptyset} = 1$.
\item[(ii)] $X$ is positive semidefinite.
\end{itemize}
Moreover, $\Etilde$ satisfies $\{p_i(x) = 0\}_{i=1}^m$ if and only if
\begin{itemize}
\item[(iii)] Let $\cU$ be the space of functions in $\cV$ of degree at most $d$ which can be written as $\sum_{i=1}^m p_i q_i$ for some $q_i \in \cV$. Then, $\ip{M_f, X_{\Etilde}} = 0$ for any $f \in \cU$. 
\end{itemize}

Hence, (\ref{eqn:polysos}) can be written as the following semidefinite program

\begin{equation}
\tag{$\mathsf{SDP}_d$}
\label{eqn:polysossdp}
\begin{aligned}
& & \max \quad & \ip{M_{f}, X}  \\
& & \text{ subject to }\quad & X_{\emptyset,\emptyset}=1 \\
& & & \ip{M_{q}, X} = 0 \text{ for all $q \in \cB$} \\
& & & X \succeq 0,
\end{aligned}
\end{equation}
where $\cB$ is any finite subset of $\cU$ which spans $\cU$, for example, $$\cB = \{x_S p_i(x) : i \in [m], |S| \leq d - \deg(p_i)\}.$$ 

\section{Proof of Theorem \ref{thm:C}}

Let $y \in \{\pm 1\}^n$ such that $\one^T y = 0$, $\sigma > 0$, and $\bW \in (\RR^n)^{\otimes 4}$ be 4-tensor with independent, standard Gaussian entries. Given a tensor $\bT = y^{\otimes 4} + \sigma \bW$, we would like to recover the planted solution $y$ from $\bT$. Let $f(x) = \ip{\bT,x^{\otimes 4}}$. The maximum-likelihood estimator is given by the optimum solution of 
$$
\max_{x\in\{\pm 1\}^n: \one^T x = 0} f(x).
$$
Consider the SoS relaxation of degree $4$
$$
\begin{aligned}
& & \max \quad & \Etilde[f] \\
& & \text{ subject to } \quad & \Etilde \text{ is a pseudo-expectation of degree 4 on $\{\pm 1\}^n$} \\
& & & \text{satisfying } \sum_{i=1}^n x_i = 0.
\end{aligned}
$$
Let $\EE_{U(\{y,-y\})}$ be the expectation operator of the uniform distribution on $\{y, -y\}$, i.e., $\EE_{U(\{y,-y\})}[x_S] = y_S$ if $|S|$ is even, and $\EE_{U(\{y,-y\})}[x_S] = 0$ if $|S|$ is odd. 

If $\EE_{U(\{y,-y\})}$ is the optimal solution of the relaxation, then we can recover $y$ from it up to a global sign flip. First we give an upper bound on $\sigma$ to achieve it with high probability.

\begin{theorem}[Part one of Theorem \ref{thm:C}]
\label{thm:C1}
Let $\bT = y^{\otimes 4} + \sigma \bW$ and $f(x) = \ip{\bT, x^{\otimes 4}}$ be as defined above. If $\sigma \lesssim \frac{n}{\sqrt{\log n}}$, then the relaxation $\max_{\Etilde} \Etilde[f]$ over pseudo-expectation $\Etilde$ of degree 4 satisfying $\sum_{i=1}^n x_i = 0$ is maximized when $\Etilde = \EE_{U(\{y,-y\})}$ with probability $1-o_n(1)$. 
\end{theorem}

We can reduce Theorem \ref{thm:C1} to the matrix version of the problem via flattening \cite{Montanari2014}. Given a 4-tensor $\bT$, the canonical flattening of $\bT$ is defined as $n^2 \times n^2$ matrix $T$ with entries $T_{(i,j),(k,\ell)} = \bT_{ijk\ell}$. Then, $T = \widetilde{y}\widetilde{y}^T + \sigma W$ where $\widetilde{y}$ is the vectorization of $yy^T$, and $W$ is the flattening of $\bW$. Note that this is an instance of $\mathbb{Z}_2$-synchronization model with Gaussian noises. It follows that with high probability the exact recovery is possible when $\sigma \lesssim \frac{n}{\sqrt{\log n}}$ (see Proposition 2.3 in \cite{Bandeira2016laplacian}). 

We complement the result by providing a lower bound on $\sigma$ which is off by polylog factor. 

\begin{theorem}[Part two of Theorem \ref{thm:C}]
\label{thm:C2}
Let $c > 0$ be a small constant. If $\sigma \geq n (\log n)^{1/2+c}$, then there exists a pseudo-expectation $\Etilde$ of degree 4 on the hypercube $\{\pm 1\}^n$ satisfying $\sum_{i=1}^n x_i = 0$ such that $\Etilde[f] > f(y)$ with probability $1-o_n(1)$. 
\end{theorem}

\subsection{Proof of Theorem \ref{thm:C2}}

Let $g(x)$ be the noise part of $f(x)$, i.e., $g(x) = \ip{\bW, x^{\otimes 4}}$. Let $\Etilde$ be a pseudo-expectation of degree 4 on the hypercube which satisfies the equality $\sum_{i=1}^n x_i = 0$. We have $\Etilde[f] \geq \sigma \Etilde[g]$ since $\Etilde[(x^T y)^4] \geq 0$.

\begin{lemma}
\label{lem:C}
Let $g(x)$ be the polynomial as defined above. Then, there exists a pseudo-expectation of degree 4 on the hypercube satisfying $\one^T x = 0$ such that
$$
\Etilde[g] \gtrsim \frac{n^3}{(\log n)^{1/2+o(1)}}.
$$
\end{lemma}

We prove Theorem \ref{thm:C2} using the lemma. 
\begin{proof}[Proof of Theorem \ref{thm:C2}]
Note that $g(y) = \ip{\bW, y^{\otimes 4}}$ is a Gaussian random variable with variance $n^4$. So, $g(y) \leq n^2 \log n$ with probability $1-o(1)$. Let $\Etilde$ be the pseudo-expectation satisfying the conditions in the lemma. Then, with probability $1-o(1)$ we have
\begin{eqnarray*}
\Etilde[f] - f(y) &=& -(n^4-\Etilde[(y^T x)^4]) + \sigma(\Etilde[g] - g(y)) \\
&\geq& -n^4 + \sigma\left(\frac{n^3}{\log n} + n^2 \log n\right) \\
&\geq& -n^4 + (1-o(1))\frac{\sigma n^3}{(\log n)^{1/2+o(1)}}.
\end{eqnarray*}
Since $\sigma > n(\log n)^{1/2+c}$ for some fixed constant $c>0$, we have $\Etilde[f] - f(y) > 0$ as desired.
\end{proof}

In the remainder of the section, we prove Lemma \ref{lem:C}.

\subsubsection{Outline}

We note that our method shares a similar idea which appears in \cite{Hopkins2015} and \cite{Bhattiprolu2016}.

We are given a random polynomial $g(x) = \ip{\bW, x^{\otimes 4}}$ where $\bW$ has independent standard Gaussian entries. We would like to construct $\Etilde = \Etilde_{\bW}$ which has large correlation with $\bW$. If we simply let 
$$
\Etilde[x_{i_1} x_{i_2} x_{i_3} x_{i_4}] = 
\frac{1}{24} \sum_{\pi \in \mathcal{S}_4} \bW_{i_{\pi(1)},i_{\pi(2)},i_{\pi(3)},i_{\pi(4)}}
$$
for $\{i_1 < i_2 < i_3 < i_4\} \subseteq [n]$ and $\Etilde[x_T]$ be zero if $|T| \leq 3$, then
$$
\Etilde[g] = \frac{1}{24} \sum_{1\leq i_1<i_2<i_3<i_4\leq n} \left(\sum_{\pi \in \mathcal{S}_4} \bW_{i_{\pi(1)},i_{\pi(2)},i_{\pi(3)},i_{\pi(4)}} \right)^2
$$
so the expectation of $\Etilde[g]$ over $\bW$ would be equal to $\binom{n}{4} \approx \frac{n^4}{24}$. However, in this case $\Etilde$ does not satisfies the equality $\one^T x = 0$ nor the conditions for pseudo-expectations.

To overcome this, we first project the $\Etilde$ constructed above to the space of linear functionals which satisfy the equality constraints ($x_i^2 = 1$ and $\one^T x = 0$). Then, we take a convex combination of the projection and a pseudo-expectation to control the spectrum of the functional. The details are following:

\begin{itemize}
\item[(1)] (Removing degeneracy) We establish the one-to-one correspondence between the collection of linear functionals on $n$-variate, even multilinear polynomials of degree at most 4 and the collection of linear functionals on $(n-1)$-variate multilinear polynomials of degree at most 4 by posing $x_n = 1$. This correspondence preserves positivity.

\item[(2)] (Description of equality constraints) Let $\psi$ be a linear functional on $(n-1)$-variate multilinear polynomials of degree at most 4. We may think $\psi$ as a vector in $\RR^{\binom{n-1}{\leq 4}}$. Then, the condition that $\psi$ satisfies $\sum_{i=1}^{n-1} x_i + 1 = 0$ can be written as $A\psi = 0$ for some matrix $A$.

\item[(3)] (Projection)
Let $w \in \RR^{\binom{n-1}{\leq 4}}$ be the coefficient vector of $g(x)$. Let $\Pi$ be the projection matrix to the space $\{x: Ax = 0\}$. In other words,
$$
\Pi = Id - A^T(AA^T)^{\dagger}A
$$
where $Id$ is the identity matrix of size $\binom{n-1}{\leq 4}$ and $(\cdot)^{\dagger}$ denotes the pseudo-inverse. Let $e$ be the first column of $\Pi$ and $\psi_1 = \frac{\Pi w}{e^T w}$. Then $(\psi_1)_\emptyset = 1$ and $A\psi_1 = 0$ by definition.

\item[(4)] (Convex combination) Let $\psi_0 = \frac{e}{e^T e}$. We note that $\psi_0$ corresponds to the expectation operator of uniform distribution on $\{x \in \{\pm 1\}^n: \one^T x = 0\}$. 

We will construct $\psi$ by
$$
\psi = (1-\epsilon)\psi_0 + \epsilon \psi_1
$$
with an appropriate constant $\epsilon$. Equivalently,
$$
\psi = \psi_0 + \frac{\epsilon}{e^T w}\cdot \left(\Pi - \frac{ee^T}{e^T e}\right) w.
$$
 
\item[(5)] (Spectrum analysis) We bound the spectrum of the functional $\left(\Pi - \frac{ee^T}{e^T e} \right) w$ to decide the size of $\epsilon$ for $\psi$ being positive semidefinite.

\end{itemize}

\subsubsection{Removing degeneracy}

Recall that
$$
g(x) = \sum_{i,j,k,l \in [n]} \bW_{ijkl} x_i x_j x_k x_\ell.
$$
Observe that $g$ is even, i.e., $g(x) = g(-x)$ for any $x \in \{\pm 1\}^n$. To maximize such an even function, we claim that we may only consider the pseudo-expectations such that whose odd moments are zero.

\begin{proposition}
Let $\Etilde$ be a pseudo-expectation of degree 4 on hypercube satisfying $\sum_{i=1}^n x_i = 0$. Let $p$ be a degree 4 multilinear polynomial which is even. Then, there exists a pseudo-expectation $\Etilde'$ of degree 4 such that $\Etilde[p]=\Etilde'[p]$ and $\Etilde'[x_S] = 0$ for any $S\subseteq [n]$ of odd size. 
\end{proposition}

\begin{proof}
Let $\Etilde$ be a pseudo-expectation of degree 4 on hypercube satisfying $\sum_{i=1}^n x_i = 0$. Let us define a linear functional $\Etilde_0$ on the space of multilinear polynomials of degree at most 4 so that $\Etilde_0[x_S] = (-1)^{|S|}\Etilde[x_S]$ for any $S \in \binom{[n]}{\leq 4}$. Then, for any multilinear polynomial $q$ of degree at most 2, we have
$$
\Etilde_0[q(x)^2] = \Etilde[q(-x)^2] \geq 0.
$$
Also, $\Etilde_0$ satisfies $\Etilde_0[1] = 1$ and
$$
\Etilde_0\left[\left(\sum_{i=1}^n x_i\right)q(x)\right] = -\Etilde\left[\left(\sum_{i=1}^n x_i\right) q(-x)\right] = 0
$$
for any $q$ of degree 3. Thus, $\Etilde_0$ is a valid pseudo-expectation of degree 4 satisfying $\sum_{i=1}^n x_i = 0$. 

Let $\Etilde' = \frac{1}{2}(\Etilde + \Etilde_0)$. This is again a valid pseudo-expectation, since the space of pseudo-expectations is convex. We have $\Etilde'[p(x)] = \Etilde[p(x)] = \Etilde_0[p(x)]$ since $p$ is even, and $\Etilde'[x_S] = (1+(-1)^{|S|})\Etilde[x_S] = 0$ for any $S$ of odd size.
\end{proof}

Let $\cE$ be the space of all pseudo-expectations of degree 4 on $n$-dimensional hypercube with zero odd moments. Let $\cE'$ be the space of all pseudo-expectations of degree 4 on $(n-1)$-dimensional hypercube. We claim that there is a bijection between two spaces.

\begin{proposition}
Let $\psi \in \cE$. Let us define a linear functional $\psi'$ on the space of $(n-1)$-variate multilinear polynomials of degree at most 4 so that for any $T \subseteq [n-1]$ with $|T|\leq 4$
$$
\psi'[x_T] = \begin{cases}
\psi[x_{T \cup\{n\}}] & \text{ if $|T|$ is odd } \\
\psi[x_T] & \text{ otherwise.}
\end{cases}
$$
Then, $\psi \mapsto \psi'$ is a bijective mapping from $\cE$ to $\cE'$. 
\end{proposition}

\begin{proof}
We say linear functional $\psi$ on the space of polynomials of degree at most $2\ell$ is \emph{positive semidefinite} if $\psi[q^2] \geq 0$ for any $q$ of degree $\ell$. 

Note that the mapping $\psi' \mapsto \psi$ where $\psi[x_S] = \psi'[x_{S \setminus \{n\}}]$ for any $S \subseteq [n]$ of even size is the inverse of $\psi \mapsto \psi'$. Hence, it is sufficient to prove that $\psi$ is positive semidefinite if and only if $\psi'$ is positive semidefinite.

\noindent ($\Rightarrow$) Let $q$ be an $n$-variate polynomial of degree at most 2. Let $q_0$ and $q_1$ be polynomials in $x_1,\cdots,x_{n-1}$ such that
$$
q(x_1,\cdots,x_n) = q_0(x_1,\cdots,x_{n-1}) + x_n q_1(x_1,\cdots,x_{n-1}).
$$ 
We get
$\psi'[q^2] = \psi'[(q_0+x_n q_1)^2] = \psi'[(q_0^2+q_1^2)+2x_nq_0 q_1]$.
For $i=1,2$, let $q_{i0}$ and $q_{i1}$ be the even part and the odd part of $q_i$, respectively. Then we have
\begin{eqnarray*}
\psi'[q^2] &=& \psi'[(q_{00}^2+q_{01}^2+q_{10}^2+q_{11}^2)+2x_n (q_{00} q_{11}+q_{01}q_{10})] \\
&=& \psi[(q_{00}^2+q_{01}^2+q_{10}^2+q_{11}^2)+2(q_{00} q_{11}+q_{01}q_{10})] \\
&=& \psi[(q_{00}+q_{11})^2 + (q_{10}+q_{01})^2] \geq 0.
\end{eqnarray*}
The first equality follows from that $\psi'[q] = 0$ for odd $q$. Hence, $\psi'$ is positive semidefinite. 

\noindent ($\Leftarrow$) Let $q$ be an $(n-1)$-variate polynomial of degree at most 2. Let $q_0$ and $q_1$ be the even part and the odd part of $q$, respectively. Then,
$$
\psi[q^2] = \psi[(q_0^2 + q_1^2) + 2q_0 q_1]. 
$$
Note that $q_0^2+q_1^2$ is even and $q_0 q_1$ is odd. So,
$$
\psi[q^2] = \psi'[(q_0^2+q_1^2) + 2x_n q_0 q_1] = \psi'[(q_0 + x_n q_1)^2] \geq 0.
$$
Hence $\psi$ is positive semidefinite. 
\end{proof}

In addition to the proposition, we note that $\psi$ satisfies $\sum_{i=1}^n x_i = 0$ if and only if $\psi'$ satisfies $1+ \sum_{i=1}^{n-1} x_i = 0$. Hence, maximizing $\Etilde[g]$ over $\Etilde \in \cE$ satisfying $\sum_{i=1}^n x_i = 0$ is equivalent to 
$$
\max_{\psi' \in \cE'} \psi'[g'] \quad \text{subject to} \quad \psi' \text{ satisfies } 1+\sum_{i=1}^{n-1} x_i = 0,
$$
where $g'(x_1,\cdots,x_{n-1}) = g(x_1,\cdots,x_{n-1},1)$. 

\subsubsection{Matrix expression of linear constraints}

Let $\cF$ be the set of linear functional on the space of $(n-1)$-variate multilinear polynomials of degree at most 4. We often regard a functional $\psi \in \cF$ as a $\binom{n-1}{\leq 4}$ dimensional vector with entries $\psi_{S} = \psi[x_S]$ where $S$ is a subset of $[n-1]$ of size at most 4. The space $\cE'$ of pseudo-expectations of degree 4 (on $(n-1)$-variate multilinear polynomials) is a convex subset of $\cF$. 

Observe that $\psi \in \cF$ satisfies $1+\sum_{i=1}^{n-1} x_i = 0$ if and only if 
$$
\psi\left[\left(1+\sum_{i=1}^{n-1} x_i\right) x_S\right] = 0
$$
for any $S \subseteq [n-1]$ with $|S| \leq 3$. 

Let $s$, $t$ and $u$ be integers such that $0 \leq s, t \leq 4$ and $0 \leq u \leq \min(s,t)$. Let $M_{s,t}^u$ be the matrix of size $\binom{n-1}{\leq 4}$ such that 
$$
(M_{s,t}^u)_{S,T} = \begin{cases}
1 & \text{ if $|S|=s$, $|T|=t$, and $|S \cap T| = u$} \\
0 & \text{ otherwise}
\end{cases}
$$
for $S, T \in \binom{[n-1]}{\leq 4}$. Then, the condition that $\psi \in \cF$ satisfying $1+\sum_{i=1}^{n-1} x_i = 0$ can be written as $A\psi = 0$ where
$$
A = M_{0,0}^0 + M_{0,1}^{0} + \sum_{s=1}^3 (M_{s,s-1}^{s-1} + M_{s,s}^s + M_{s,s+1}^s).
$$

\subsubsection{Algebra generated by $M_{s,t}^u$}

Let $m$ be a positive integer greater than 8. For nonnegative integers $s,t,u$, let $M_{s,t}^u$ be the $\binom{m}{\leq 4} \times \binom{m}{\leq 4}$ matrix with
$$
(M_{s,t}^u)_{S,T} = \begin{cases}
1 & \text{ if $|S|=s$, $|T|=t$, and $|S\cap T|=u$}\\
0 & \text{ otherwise,}
\end{cases}
$$
for $S,T \subseteq [m]$ with $|S|,|T| \leq 4$. Let $\cA$ be the algebra of matrices
$$
\sum_{0 \leq s,t\leq 4} \sum_{u=0}^{s \wedge t} x_{s,t}^u M_{s,t}^u
$$
with complex numbers $x_{s,t}^u$. This algebra $\cA$ is a $C^*$-algebra: it is a complex algebra which is closed under taking complex conjugate. $\cA$ is a subalgebra of the Terwilliger algebra of the Hamming cube $H(m,2)$ \cite{terwilliger1992subconstituent}, \cite{schrijver2005new}. 

Note that $\cA$ has dimension $55$ which is the number of triples $(s,t,u)$ with $0\leq s,t \leq 4$ and $0 \leq u \leq s \wedge t$. 

Define
$$
\beta_{s,t,r}^u := \sum_{p=0}^{s \wedge t} (-1)^{p-t} \binom{p}{u} \binom{m-2r}{p-r} \binom{m-r-p}{s-p} \binom{m-r-p}{t-p}
$$
for $0 \leq s,t \leq 4$ and $0 \leq r,u \leq s \wedge t$.
The following theorem says that matrices in the algebra $\cA$ can be written in a block-diagonal form with small sized blocks. 

\begin{theorem}[\cite{schrijver2005new}]
\label{thm:blkdiag}
There exists an orthogonal $\binom{m}{\leq 4} \times \binom{m}{\leq 4}$ matrix $U$ such that for $M \in \cA$ with 
$$
M = \sum_{s,t=0}^4\sum_{u=0}^{s\wedge t} x_{s,t}^u M_{s,t}^u,
$$
the matrix $U^T M U$ is equal to the matrix
$$
\begin{pmatrix}
C_0 & 0 & 0 & 0 & 0 \\
0 & C_1 & 0 & 0 & 0 \\
0 & 0 & C_2 & 0 & 0 \\
0 & 0 & 0 & C_3 & 0 \\
0 & 0 & 0 & 0 & C_4
\end{pmatrix}
$$
where each $C_r$ is a block diagonal matrix with $\binom{m}{r} - \binom{m}{r-1}$ repeated, identical blocks of order $5-r$:
$$
C_r = \begin{pmatrix}
B_r & 0 & \cdots & 0 \\
0 & B_r & \cdots & 0 \\
\vdots & \vdots & \ddots & \vdots \\
0 & 0 & \cdots & B_r
\end{pmatrix},
$$
and 
$$
B_r = \left(\sum_{u} \binom{m-2r}{s-r}^{-1/2}\binom{m-2r}{t-r}^{-1/2} \beta_{s,t,r}^u x_{s,t}^u\right)_{s,t=r}^{4}.
$$
\end{theorem}

For brevity, let us denote this block-diagonalization of $M$ by the tuple of matrices $(B_0,B_1,B_2,B_3,B_4)$ with $(5-r) \times (5-r)$ matrix $B_r$'s. 

\subsubsection{Projection}

Recall that $g'(x_1,\cdots,x_{n-1})$ is equal to
$$
\sum_{i,j,k,\ell \in [n]} \bW_{ijk\ell} x_{\{i,j,k,\ell\}\setminus \{n\}}.
$$
Let $c \in \RR^{\binom{n-1}{\leq 4}}$ be the coefficient vector of $g'$. Since entries of $\bW$ are independent Standard Gaussians, $c$ is a Gaussian vector with a diagonal covariance matrix $\Sigma = \EE[cc^T]$ where
$$
\Sigma_{S,S} = \begin{cases}
n & \text{ if $|S| = 0$} \\
12n-16 & \text{ if $|S|=1$ or $|S|=2$} \\
24 & \text{ if $|S| = 3$ or $|S|=4$.}
\end{cases}
$$
Let $w = \Sigma^{-1/2} c$. By definition, $w$ is a random vector with i.i.d. standard normal entries. 

Let $\Pi = Id - A^T (AA^T)^{+} A$ where $(AA^T)^+$ is the Moore-Penrose pseudoinverse of $AA^T$ and $Id$ is the identity matrix of order $\binom{n-1}{\leq 4}$. Then, $\Pi$ is the orthogonal projection matrix onto the nullspace of $A$. Since $A$, $A^T$ and $Id$ are all in the algebra $\cA$, the projection matrix $\Pi$ is also in $\cA$.

Let $e$ be the first column of $\Pi$ and  
$$
\psi_0 := \frac{e}{e^T e} \quad \text{and} \quad \psi_1 := \frac{\Pi w}{e^T w}. 
$$
We have $A\psi_0 = A\psi_1 = 0$ by definition of $\Pi$, and $(\psi_0)_\emptyset = (\psi_1)_\emptyset = 1$ since $(\Pi w)_\emptyset = e^T w$.

Let $\epsilon$ be a real number with $0 < \epsilon < 1$ and $\psi = (1-\epsilon)\psi_0 + \epsilon \psi_1$. This functional still satisfies $A\psi = 0$ and $\psi_\emptyset = 1$, regardless of the choice of $\epsilon$. We would like to choose $\epsilon$ such that $\psi$ is positive semidefinite with high probability.

\subsubsection{Spectrum of $\psi$}

Consider the functional $\psi_0 = \frac{e}{e^T e}$. It has entries
$$
(\psi_0)_S = \begin{cases}
1 & \text{ if $S = \emptyset$} \\
-\frac{1}{n-1} & \text{ if $|S|=1$ or 2} \\
\frac{3}{(n-1)(n-3)} & \text{ if $|S|=3$ or 4},
\end{cases}
$$
for $S \subseteq [n-1]$ of size at most 4. We note that this functional corresponds to the degree 4 or less moments of the uniform distribution on the set of vectors $x \in \{\pm 1\}^{n-1}$ satisfying $\sum_{i=1}^{n-1} x_i + 1 = 0$.

\begin{proposition}
Let $\psi$ be a vector in $\RR^{\binom{n-1}{\leq 4}}$ such that $A\psi = 0$ and $p$ be an $(n-1)$-variate multilinear polynomial of degree at most 2. Suppose that $\psi_0[p^2] = 0$. Then, $\psi[p^2] = 0$. 
\end{proposition}

\begin{proof}
Let $\cU = \{x \in \{\pm 1\}^{n-1}: \sum_{i=1}^{n-1} x_i + 1 = 0\}$. Note that $\psi_0$ is the expectation functional of the uniform distribution on $\cU$ as we seen above. Hence, $\psi_0[p^2] = 0$ if and only if $p(x)^2 = 0$ for any $x \in \cU$. 

On the other hand, the functional $\psi$ is a linear combination of functionals $\{p \mapsto p(x): x\in \cU\}$ since $A\psi = 0$.  Hence, if $\psi_0[p^2] = 0$ then $\psi[p^2] = 0$ as $p(x)^2 = 0$ for any $x \in \cU$.
\end{proof}

Recall that $\psi = (1-\epsilon)\psi_0 + \epsilon \psi_1$ where
$\psi_0 = \frac{e}{e^Te}$ and $\psi_1 = \frac{\Pi w}{e^T w}$. Let $\psi_1' = e^T w \cdot (\psi_1 - \psi_0)$. Then,
\begin{eqnarray*}
\psi_1' &=& \Pi w - \frac{e^T w}{e^T e} e \\
&=& \left(\Pi - \frac{ee^T}{e^Te}\right) w
\end{eqnarray*}
and $\psi = \psi_0 + \frac{\epsilon}{e^T w} \psi_1'$. We note that $A\psi_1' = 0$ since $\psi_1'$ is a linear combination of $\psi_0$ and $\psi_1$.

Let $X_{\psi_0}$ and $X_{\psi_1'}$ be the moment matrix of $\psi_0$ and $\psi_1'$ respectively. Let $X_{\psi}$ be the moment matrix of $\psi$. Clearly,
$$
X_{\psi} = X_{\psi_0} + \frac{\epsilon}{e^T w} X_{\psi_1'}.
$$
Moreover, for any $p \in \RR^{\binom{n-1}{\leq 2}}$ satisfying $X_{\psi_0} p = 0$, we have $X_{\psi_1'}p = 0$ by the proposition. Hence, $X_{\psi} \succeq 0$ if
$$
\frac{\epsilon}{|e^T w|} \|X_{\psi_1'}\| \leq \lambda_{\min,\neq 0}(X_{\psi_0})
$$
where $\lambda_{\min,\neq 0}$ denotes the minimum nonzero eigenvalue. 

We note that $e^T w$ and $\| X_{\psi_1'}\|$ are independent random variables. It follows from that $w$ is a gaussian vector with i.i.d. standard entries, and that $e$ and $\left(\Pi - \frac{ee^T}{e^T e}\right)$ are orthogonal. Hence, we can safely bound $e^T w$ and $\| X_{\psi_1'}\|$ separately. 

To bound $\|X_{\psi_1'}\|$ we need the following theorem.

\begin{theorem}[Matrix Gaussian (\cite{tropp2012user})]
\label{thm:matgauss}
Let $\{A_k\}$ be a finite sequence of fixed, symmetric matrices with dimension $d$, and let $\{\xi_k\}$ be a finite sequence of independent standard normal random variables. Then, for any $t \geq 0$,
$$
\Pr\left[\left\|\sum_k \xi_k A_k \right\| \geq t\right] \leq d \cdot e^{-t^2/2\sigma^2} \quad \text{where} \quad \sigma^2 := \left\|\sum_k A_k^2 \right\|.
$$
\end{theorem}

For each $U \subseteq [n-1]$ with size at most 4, let $Y_U$ be the $\binom{n-1}{\leq 2} \times \binom{n-1}{\leq 2}$ matrix with entries
$$
(Y_U)_{S,T} = \begin{cases}
1 & \text{ if $S \oplus T = U$ } \\
0 & \text{ otherwise.}
\end{cases}
$$
We can write $X_{\psi_1'}$ as
$$
X_{\psi_1'} = \sum_{\substack{U \subseteq [n-1]\\ |U|\leq 4}} (\psi_1')_U Y_U.
$$
Since $\psi_1' = \left(\Pi - \frac{ee^T}{e^T e}\right) w$, we have
\begin{eqnarray*}
X_{\psi_1'} &=& \sum_{\substack{U \subseteq [n-1]\\ |U|\leq 4}} \sum_{\substack{V \subseteq [n-1]\\ |V|\leq 4}} \left(\Pi - \frac{ee^T}{e^Te}\right)_{U,V} w_V Y_U \\
&=& \sum_V w_V \left(\sum_U \left(\Pi - \frac{ee^T}{e^Te}\right)_{U,V} Y_U\right).
\end{eqnarray*}
By Theorem \ref{thm:matgauss}, $\|X_{\psi_1'}\|$ is roughly bounded by $(\|\Sigma_X\|\log n)^{1/2}$ where 
$$
\Sigma_X := \sum_V \left(\sum_U \left(\Pi - \frac{ee^T}{e^Te}\right)_{U,V} Y_U\right)^2.
$$

\begin{proposition}
For each $I, J \in \binom{[n-1]}{\leq 2}$, the $(I,J)$ entry of $\Sigma_X$ only depends on $|I|$, $|J|$ and $|I \cap J|$, i.e., $\Sigma_X$ is in the algebra $\cA$. 
\end{proposition}

\begin{proof}
Note that 
\begin{eqnarray*}
\Sigma_X &=& \sum_V \sum_{U_1, U_2} \left(\Pi - \frac{ee^T}{e^T e}\right)_{U_1,V} \left(\Pi - \frac{ee^T}{e^T e}\right)_{V,U_2} Y_{U_1} Y_{U_2} \\
&=& \sum_{U_1, U_2} \left(\left(\Pi - \frac{ee^T}{e^T e}\right)^2\right)_{U_1,U_2} Y_{U_1} Y_{U_2} \\
&=& \sum_{U_1,U_2} \left(\Pi - \frac{ee^T}{e^T e}\right)_{U_1,U_2} Y_{U_1} Y_{U_2}.
\end{eqnarray*}
Hence,
$$
(\Sigma_X)_{I,J} = \sum_{K \in \binom{[n-1]}{\leq 2}} \left(\Pi - \frac{ee^T}{e^T e}\right)_{I \oplus K, J\oplus K},
$$
which is invariant under any permutation $\pi$ on $[n-1]$ as $\Pi - \frac{ee^T}{e^T e}$ is. It implies that $\Sigma_X \in \cA$. 
\end{proof}

The block-diagonalization of $\Sigma_X$ is $(u_0u_0^T, u_1u_1^T, u_2 u_2^T, 0, 0)$ where 
\begin{eqnarray*}
u_0 &=& \sqrt{\frac{n(n-3)(n-5)}{3n-14}} \begin{bmatrix}
1 & -\frac{1}{\sqrt{n-1}} & -\sqrt{\frac{n-2}{2(n-1)}} & 0 & 0\end{bmatrix}^T \\
u_1 &=& \sqrt{\frac{(n-6)(3n^4-24n^3+59n^2-66n+32)}{2(n-1)(n-2)(3n-14)}} \begin{bmatrix}
1 & -\frac{1}{\sqrt{n-3}} & 0 & 0
\end{bmatrix}^T \\
u_2 &=& \sqrt{\frac{(n-6)(3n^4-24n^3+59n^2-66n+32)}{2(n-1)(n-3)(3n-14)}} \begin{bmatrix} 1 & 0 & 0 \end{bmatrix}^T.
\end{eqnarray*}

Hence, $\|\Sigma_X\|$ is equal to the maximum of $\|u_i\|^2$ among $i=0,1,2$, which is at most $(1/2+o_n(1))n^2$. We get the following result:

\begin{proposition}
If $\epsilon < o_n\left(\frac{1}{n\sqrt{\log n}}\right)$, then with probability $1-o_n(1)$ the moment matrix $X_\psi$ is positive-semidefinite.
\end{proposition}

\begin{proof}
By theorem \ref{thm:matgauss}, we have
$$
\Pr\left(\|X_{\psi_1'}\| \geq t \right) \leq \binom{n-1}{\leq 2} \cdot e^{-t^2/2\|\Sigma_X\|}.
$$
Let $t = 3n\sqrt{\log n}$. Since $\|\Sigma_X\| \leq (1/2+o(1))n^2$, we have $\|X_{\psi_1'}\| \leq 3n\sqrt{\log n}$ with probability $1-n^{-\Omega(1)}$. On the other hand, note that
$$
\Pr\left(|e^T w| \leq t \right) \leq \frac{t}{\sqrt{2\pi}}.
$$
It implies that $|e^T w| > g(n)$ with probability $1-o_n(1)$ for any $g(n) = o_1(1)$. Thus,
$$
\frac{\|X_{\psi_1'}\|}{|e^T w|}  \lesssim \frac{n\sqrt{\log n}}{g(n)}
$$
almost asymptotically surely. Together with the fact that $\lambda_{\min, \neq 0}(X_{\psi_0}) = 1-o_n(1)$, we have $X_{\psi} \succeq 0$ whenever $\epsilon < \frac{g(n)}{n\sqrt{\log n}}$ for some $g(n) = o_n(1)$. 

\end{proof}

\subsubsection{Putting it all together}

We have constructed a linear functional $\psi$ on the space of $(n-1)$-variate multilinear polynomials of degree at most 4, which satisfies (i) $\psi[1] = 1$, (ii) $\psi$ satisfies $\sum_{i=1}^{n-1} x_i + 1 = 0$, and (iii) $\psi[p^2] \geq 0$ for any $p$ of degree 2. It implies that $\psi$ is a valid pseudo-expectation of degree 4. 

Now, let us compute the evaluation of
$$
g'(x) = \sum_{i,j,k,\ell \in [n]} \bW_{ijk\ell} x_{\{i,j,k,\ell\} \setminus \{n\}}
$$
by the functional $\psi$. Recall that $c$ is the coefficient vector of $g'$ and $w = \Sigma^{-1/2} c$ where $\Sigma = \EE[cc^T]$. We have
\begin{eqnarray*}
\psi[g'] = c^T \psi &=& w^T \Sigma^{1/2} \left(\frac{e}{e^T e} + \frac{\epsilon}{e^T w} \left(\Pi - \frac{ee^T}{e^T e}\right)w\right) \\
&=& \frac{e^T \Sigma^{1/2} w}{e^T e} + \epsilon \cdot \frac{w^T \Sigma^{1/2} \left(\Pi - \frac{ee^T}{e^T e}\right) w}{e^T w}.
\end{eqnarray*}

Note that 
\begin{eqnarray*}
\EE\left[w^T \Sigma^{1/2}\left(\Pi-\frac{ee^T}{e^Te}\right) w\right] &=& \ip{\Sigma^{1/2}\left(\Pi-\frac{ee^T}{e^Te}\right), \EE[ww^T]} \\
&=& \tr\left(\Sigma^{1/2}\left(\Pi-\frac{ee^T}{e^Te}\right)\right),
\end{eqnarray*}
which is at least $(\sqrt{6}/12-o_n(1)) n^4$. Also, $|e^T w| = O(1)$ and $|e^T \Sigma^{1/2} w| = O(n)$ with high probability. Hence, with probability $1-o_n(1)$, we have
$$
\psi[g'] \gtrsim O(n) + \frac{n^4}{n(\log n)^{1/2+o(1)}} \gtrsim \frac{n^3}{(\log n)^{1/2+o(1)}}.
$$

\end{document}